\documentclass[preprint,12pt,authoryear]{elsarticle}

\usepackage{amssymb}
\usepackage{amsthm}
\usepackage{amsmath}
\usepackage{color}

\newtheorem{Theorem}{Theorem}

\newtheorem{Example}{Example}
\newtheorem{Proposition}{Proposition}

\newtheorem{Lemma}{Lemma}
\newtheorem{Definition}{Definition}
\newtheorem{Remark}{Remark}
\newtheorem{Conjecture}{Conjecture}

\usepackage{comment}
\newcommand{\beq}{\begin{equation}}
\newcommand{\eeq}{\end{equation}}
\usepackage{hyperref}

\DeclareMathOperator{\Res}{Res}
\DeclareMathOperator{\PRes}{PRes}
\DeclareMathOperator{\CRes}{GCP}

\newcommand{\bx}{\mathbf{x}}
\newcommand{\by}{\mathbf{y}}
\newcommand{\bp}{\mathbf{p}}
\newcommand{\bq}{\mathbf{q}}
\newcommand{\ba}{\mathbf{a}}

\newcommand{\KK}{\mathbb{K}}
\newcommand{\PP}{\mathbb{P}}
\newcommand{\AAA}{\mathbb{A}}

\newcommand{\Ra}[1]{#1}
\newcommand{\Rb}[1]{#1}

\journal{}

\begin{document}

\begin{frontmatter}



\title{Persistent components in\\ \Ra{Canny's Generalized Characteristic Polynomial}\\ \emph{\normalsize In memory of Agnes Szanto}}


\author[1]{Gleb Pogudin}
\ead{gleb.pogudin@polytechnique.edu}

\affiliation[1]{organization={LIX, CNRS, Ecole polytechnique, Institute Polytechnique de Paris},
    city={Paris},
    country={France}}

\begin{abstract}
  When using resultants for elimination, one standard issue is that the resultant vanishes if the variety contains components of dimension larger than the expected dimension.
  J. Canny proposed an elegant construction, generalized characteristic polynomial, to address this issue by symbolically perturbing the system before the resultant computation.
  Such perturbed resultant would typically involve artefact components only loosely related to the geometry of the variety of interest.
  For removing these components, J.M. Rojas proposed to take the greatest common divisor of the results of two different perturbations.
  In this paper, we investigate this construction, and show that the extra components persistent under taking different perturbations must come either from singularities or from positive-dimensional fibers.
\end{abstract}

\begin{keyword}
resultant \sep generalized characteristic polynomial \sep syzygies



\end{keyword}

\end{frontmatter}


\section{Introduction}

Resultants is a classical tool for performing elimination of variables in polynomial systems.
While they are often very successful in practice, their applicability may be hindered by so-called excess components.
For example, assume that we would like to find isolated zeroes of the following two-dimensional system:
\begin{equation}\label{eq:intro_ex}
f_1 := x^2 - y^2 + x - y = 0, \quad f_2 := 2 x^2 - xy - y^2 = 0.
\end{equation}
One natural approach would be to compute the resultant of $f_1$ and $f_2$ with respect to one of the variables, say $x$, in order to obtain a univariate polynomial vanishing at the $y$-coordinates of the roots of the system.
However, both the resultant $\Res_x(f_1, f_2)$ with respect to $x$ and the resultant $\Res_y(f_1, f_2)$ with respect to $y$ vanish identically and, thus, do not provide useful information about the isolated roots.
This is because the variety defined by $f_1 = f_2 = 0$ has a one-dimensional component $x - y = 0$ as can be shown by factoring $f_1$ and $f_2$ as follows:
\[
f_1 = (x - y)(x + y + 1), \quad f_2 = (x - y)(2x + y).
\]
\cite{Canny90} proposed an elegant way to circumvent this limitation. 
The idea is to ``symbolically perturb'' the system by introducing a new variable $\varepsilon$ and then considering
\begin{equation}\label{eq:preturb}
\Tilde{f}_1 := f_1 + \varepsilon x^2, \quad \Tilde{f}_2 := f_2 + \varepsilon.
\end{equation}
Now the resultant is nonzero
\begin{equation}\label{eq:res1}
\Res_x(\Tilde{f}_1, \Tilde{f}_2) = \underbrace{(-3 y^4 - 6 y^3 + 2 y^2 + 5 y + 2)}_{=: \widetilde{R}(y) } \varepsilon + (y^4 + 5 y + 1) \varepsilon^2 + \mathcal{O}(\varepsilon^3)
\end{equation}
and its lowest-order coefficient with respect to $\varepsilon$, $\widetilde{R}(y) := -3 y^4 - 6 y^3 + 2 y^2 + 5 y + 2$, is guaranteed~\citep[Theorem~3.2]{Canny90} to vanish at the $y$-coordinate of every isolated root of the system $f_1 = f_2 = 0$.
In this case, the only isolated root is $(1, -2)$ and, indeed, $y = -2$ is a root of $\widetilde{R}(y)$.
This construction has been used, for example, for computing triangular decompositions of algebraic varieties by \cite{SzantoPaper, SzantoThesis}  and for polynomial system solving by~\cite{Rojas1999}.

This beautiful construction has one potential inconvenience that $y = -2$ is not the \emph{only} root of $\widetilde{R}(y)$, there are three more roots $y = 1$ and $y = -\frac{1}{2} \pm \frac{\sqrt{3}i}{6}$.
One can think of these roots as the remainders of the one-dimensional component, see~\citep[Lemma 5.4]{Rojas1999} for a more precise statement.
In many contexts, though not always~\citep{Rojas1999}, these roots are of no interest, so it is desirable to eliminate them or maybe even not compute at the first place, and this is the question investigated in this paper.

More precisely, we will study the  \emph{double perturbation} approach proposed in~\cite[Section~5.7]{Rojas1999}.
The starting point is to observe~\citep[Section~2.2]{Rojas1999} that the original Canny's argument from~\citep{Canny90} applies to a broad general class of perturbations, not only to the one given in~\eqref{eq:preturb}.
Therefore, one can take \emph{two different generic perturbations} and compute the GCD of the corresponding resultants.
For our running example, we could take, in addition to~\eqref{eq:preturb},
\begin{equation}\label{eq:perturbed_hat}
  \hat{f}_1 := f_1 + \varepsilon (x + 1)^2, \quad \hat{f}_2 := f_2 + \varepsilon (x - 1)^2
\end{equation}
and obtain
\begin{equation}\label{eq:res_hat}
\Res_{\Rb{x}}(\hat{f}_1, \hat{f}_2) = \underbrace{(-y^4 - 11 y^3 - 21 y^2 - 5 y + 2)}_{=:\widehat{R}(\Rb{y})}\varepsilon + \mathcal{O}(\varepsilon^2).
\end{equation}
Then their greatest common divisor, $\operatorname{gcd}(\widetilde{R}(y),\; \widehat{R}(y))$, is equal to $y + 2$ and, thus, vanishes exactly at the $y$-coordinates of the isolated roots of the system, as desired.
\cite{Rojas1999} conjectured that ``we should be able to pick out these isolated roots simply by computing the gcd''.
As examples in Section~\ref{sec:examples} show, the situation is more subtle: ``persistent'' roots of Canny's resultant may also come from points with positive-dimensional fibers (Example~\ref{ex:positive}) or singularities (Examples~\ref{ex:embedded} and~\ref{ex:twisted}). 

The main result of this paper (Theorem~\ref{thm:main}) shows that these two potential sources of additional ``persistent'' components of Canny's resultant are the only ones (\Ra{thus excluding, for example, the critical points of the projection map}).
We also fully characterize such ``persistent'' components for the planar case (Proposition~\ref{prop:planar}).
The main ingredient in the proof is a connection we establish between ``persistent'' components and zero sets of the syzygy module of $f_i$'s (see Section~\ref{sec:proofs}).

The rest of the paper is organized as follows. 
We give precise definitions and state the problem formally in Section~\ref{sec:preliminaries}.
We state our main results in Section~\ref{sec:results} and illustrate them with examples in Section~\ref{sec:examples}.
The proofs are contained in Section~\ref{sec:proofs}.

\paragraph{Acknowledgements} I have learned about Canny's resultant (among many other things) from Agnes Szanto when we were working on our paper~\citep{our}, and without her encouragement none of my ideas about them would have ever been written down.

I would also like to thank Carlos D'Andrea and J. Maurice Rojas for very helpful discussions and the organizers of the FOCM'23 session dedicated to Agnes where these discussions took place. 
I would like to thank the referees for careful reading and numerous helpful comments.
This work has partly been supported by the French ANR-22-CE48-0008 OCCAM project.

\section{Formal setup and problem statement}\label{sec:preliminaries}

Throughout the paper, all fields are assumed to be of characteristic zero.

\begin{Definition}[Resultant]
    Let $f_1, \ldots, f_n \in \mathbb{K}[\bx]$ be $n$ homogeneous polynomials in $\bx = (x_1, \ldots, x_n)$.
    By $\Res_{\bx}(f_1, \ldots, f_n)$ we will denote the \emph{resultant} of $f_1, \ldots, f_n$ with respect to $\bx$, that is, the unique (up to a constant factor) irreducible polynomial in the coefficients of $f_1, \ldots, f_n$ with the property
    \[
      \Res_{\bx}(f_1, \ldots, f_n) = 0 \quad \iff \quad \exists \bx^\ast \in \overline{\mathbb{K}}^n\colon \bx^\ast \neq \mathbf{0} \;\&\; f_1(\bx^\ast) = \ldots = f_n(\bx^\ast) = 0,
    \]
    where $\overline{\mathbb{K}}$ denotes the algebraic closure of $\mathbb{K}$.
    For the existence and basic properties of the resultant, we refer to~\cite[Chapter~3, \S2]{CLO}.
\end{Definition}

We will typically compute resultants for polynomials $f_1, \ldots, f_n \in \mathbb{K}[\bx, \by]$ homogeneous with respect to $\bx = (x_1, \ldots, x_n)$ but not necessarily with respect to $\by = (y_1, \ldots, y_m)$ (in other words, $f_1, \ldots, f_n$ are regular functions on $\mathbb{P}^{n - 1} \times \mathbb{A}^m$).
In this case, $\Res_{\bx}(f_1, \ldots, f_n)$ is a polynomial in $\by$ defining the image of the projection of $V(f_1, \ldots, f_n) \subset \mathbb{P}^{n - 1} \times \mathbb{A}^m$ to $\mathbb{A}^m$.

Now we will formally define the ``symbolic perturbaion'' of the resultant we described informally in the introduction.
This particular construction extends the generalized characteristic polynomial introduced by~\cite{Canny90} and can be viewed as a specialization of the toric generalized characteristic polynomial introduced by~\cite{Rojas1999}.

\begin{Definition}[Generalized characteristic polynomial]\label{def:canny}
    Consider polynomials $f_1, \ldots, f_n \in \mathbb{K}[\bx, \by]$ homogeneous with respect to $\bx = (x_1, \ldots, x_n)$.
    We define $d_i := \deg_{\bx} f_i$ for every $1 \leqslant i \leqslant n$ and introduce an extra variable $\varepsilon$.
    Let $\bp \in (\mathbb{K}[\bx])^n$ be a vector \Ra{of} homogeneous polynomials with $\deg_{\bx} p_i = d_i$ for $1 \leqslant i \leqslant n$ such that the only solution of the system $\bp = \mathbf{0}$ over $\overline{\mathbb{K}}$ is $\mathbf{0}$ (we will call such vector \emph{admissible}).
    We set $\widetilde{f}_i := f_i + \varepsilon p_i$ for $1 \leqslant i \leqslant n$.
    
    Then there exist~\cite[Theorem~2.4]{Rojas1999} a nonnegative integer $s$ and nonzero polynomial $R_s(\by) \in \mathbb{K}[\by]$ such that
    \[
    \Res_\bx(\widetilde{f}_1, \ldots, \widetilde{f}_n) = R_s(\mathbf{y}) \varepsilon^s + \mathcal{O}(\varepsilon^{s + 1}).
    \]
    We will call $R_s(\by)$ the \emph{generalized characteristic polynomial} and denote it by $\CRes_{\bx, \bp}(f_1, \ldots, f_n)$.
\end{Definition}

Let $X := V_{\overline{\mathbb{K}}}(f_1, \ldots, f_n)$ be the zero set of $f_1 = \ldots = f_n = 0$ over $\overline{\KK}$ and consider the projection $\pi \colon \mathbb{P}^{n - 1}_{\overline{\mathbb{K}}} \times \mathbb{A}^m_{\overline{\mathbb{K}}} \to \mathbb{A}^m_{\overline{\mathbb{K}}}$.
A component $C$ of $X$ of the lowest possible dimension $m - 1$ will be called \emph{proper}, other components will be called \emph{excess} components.
The key property of the generalized characteristic polynomial (see~\citep[Theorem~3.2]{Canny90} and~\citep[Theorem~2.4]{Rojas1999}) is the following.
\begin{Proposition}\label{prop:key_prop}
    In the notation above, let $C \subset X$ be a proper component.
    Then $\CRes_{\bx, \bp}(f_1, \ldots, f_n)$ vanishes on $\pi(C)$ for every admissible $\bp$.
\end{Proposition}

\begin{Example}
    We will revisit the system~\eqref{eq:intro_ex}. 
    Replacing $f_1$ and $f_2$ with their homogenizations {\Ra{with respect to}} $x$, we obtain:
    \[
    f_1 := x_1^2 - x_2^2 y^2 + x_1 x_2 - x_2^2y, \quad f_2 := 2 x_1^2 - x_1x_2y - x_2^2y^2.
    \]
    Taking $\bp = (x_1^2, x_2^2)$, we will obtain $\widetilde{f}_1 = f_1 + \varepsilon x_1^2,\; \widetilde{f}_2 = f_2 + \varepsilon x_2^2$ which are exactly the $\widetilde{f}_1$ and $\widetilde{f}_2$ from~\eqref{eq:preturb} after homogenization.
    Therefore, by~\eqref{eq:res1}, we have 
    \[
    \CRes_{\bx, \bp}(f_1, f_2) = -3 y^4 - 6 y^3 + 2 y^2 + 5 y + 2.
    \]
    Likewise, if we consider $\bp = ((x_ 1 + x_2)^2, (x_1 - x_2)^2)$, then we will obtain the homogenizations of $\hat{f}_1$ and $\hat{f}_2$ from~\eqref{eq:perturbed_hat} and, thus, obtain
    \[
    \CRes_{\bx, \bp}(f_1, f_2) = -y^4 - 11 y^3 - 21 y^2 - 5 y + 2.
    \]
\end{Example}

The following lemma formalizes the idea of taking the GCD of two generic GCP's.

\begin{Lemma}\label{lem:double_pert}
    Consider polynomials $f_1, \ldots, f_n \in \mathbb{K}[\bx, \by]$ homogeneous with respect to $\bx = (x_1, \ldots, x_n)$, and denote $d_i = \deg_{\bx}f_i$ for $1 \leqslant i \leqslant n$.
    Let $\Ra{\mathcal{P}}$ be the space of all vectors $\bp \in (\mathbb{K}[\bx])^n$ of homogeneous polynomials with $\deg p_i = d_i$.
    There exists a nonempty Zariski open subset $U \subset \Ra{\mathcal{P}}\times \Ra{\mathcal{P}}$ and a polynomial $R(\by) \in \mathbb{K}[\by]$ such that, for every $(\bp, \bq) \in U$, we have:
    \[
    \operatorname{GCD}(\CRes_{\bx, \bp}(f_1, \ldots, f_n), \; \CRes_{\bx, \bq}(f_1, \ldots, f_n)) = R(\by).
    \]
\end{Lemma}

\begin{proof}
    We introduce two sets of indeterminates $\ba_1$ and $\ba_2$, of cardinality $\dim \Ra{\mathcal{P}}$ each.
    We choose a basis in $\Ra{\mathcal{P}}$, and let $\bp_1$ and $\bp_2$ be vectors of homogeneous polynomials in $(\KK(\ba_1, \ba_2)[\bx])^n$ with $\deg_\bx p_{1, i} = \deg_{\bx} p_{2, i} = d_i$ for every $1 \leqslant i \leqslant n$ with the coordinates in the chosen bases being $\ba_1$ and $\ba_2$, respectively.
    Let $R_1(\ba_1, \by) \in \KK[\ba_1, \by]$ the the lowest $\varepsilon$-degree term in
    $\Res_{\bx}(f_1 + \varepsilon p_{1, 1}, \ldots, f_n + \varepsilon p_{1, n})$.
    Then $R_1(\ba_2, \by) \in \KK[\ba_2, \by]$ is the analogous term in $\Res_{\bx}(f_1 + \varepsilon p_{2, 1}, \ldots, f_n + \varepsilon p_{2, n})$.
    Let $R_0 = \operatorname{GCD}(R(\ba_1, \by), R(\ba_2, \by))$.
    Then $R_0 \in \KK[\by]$, so, 
    for every specialization of $\ba_1$ and $\ba_2$ in $\KK$ such that none of $R_1(\ba_1, \by)$ and $R_1(\ba_2, \by)$ becomes identically zero, the GCD of the specializations is divisible by $R_0$.
    Consider the points $(\ba_1^\ast, \ba_2^\ast) \in \Ra{\mathcal{P}} \times \Ra{\mathcal{P}}$ such that 
    \[
    R_1(\ba_1^\ast, \by)R_1(\ba_2^\ast, \by) \neq 0 \quad \text{ and } \deg_{\by} \operatorname{GCD}(R_1(\ba_1^\ast, \by),\; R_1(\ba_2^\ast, \by)) = \deg_{\by} R_0.
    \]
    Since these conditions can be described by inequations (for the latter, e.g., by setting the intermediate results in the Euclidean algorithm to be nonzero), one can verify that this set contains an open nonempty subset $U \subset \Ra{\mathcal{P}}\times \Ra{\mathcal{P}}$.
    Therefore, we can take $R(\by) = R_0(\by)$ and this $U$.
\end{proof}

Lemma~\ref{lem:double_pert} motivates the following main definition of this paper (which can be viewed as a specialization of the \emph{double toric perturbation} introduced by~\cite{Rojas1999}).

\begin{Definition}[Perturbed resultant]
    In the notation of Lemma~\ref{lem:double_pert}, the polynomial $R(\by)$ will be called \emph{the perturbed resultant} of $f_1, \ldots, f_n$ with respect to $\bx$ and denoted by $\PRes_{\bx}(f_1, \ldots, f_n)$.
\end{Definition}

\begin{Remark}[Computing perturbed resultant]\label{rem:computation}
    Based on the proof of Lemma~\ref{lem:double_pert}, one way to compute $\PRes$ is to compute $\CRes$ with $\bp$ being a vector of polynomials with undetermined coefficients, and then take the factors with coefficients in the ground field $\KK$.
    We used this approach for the examples in Section~\ref{sec:examples}.
    In practice, one can take two random $\bp$'s and compute the GCD as in Lemma~\ref{lem:double_pert}.
    Probability analysis of such an approach is beyond the scope of the present paper.
\end{Remark}

Proposition~\ref{prop:key_prop} implies that the perturbed resultant vanishes on the projection of every proper component whose projection has the same dimension.
One may expect that the zero set of the perturbed resultant will be exactly the union of projections of proper components (see~\cite[Section~5.7]{Rojas1999}).
The examples from Section~\ref{sec:examples} show that the situation is more subtle.
Therefore, the \textbf{main question} studied in this paper is:

\begin{quote}
    \emph{What can be said about the zero set of the perturbed resultant in terms of the zero set of the original system?}
\end{quote}

\section{Main results}\label{sec:results}

The following theorem says that the perturbed resultant may vanish at a point $P$ only in three cases:
\begin{itemize}
    \item $P$ belongs to a projection of a proper component;
    \item one of the preimages of $P$ is a singularity (see Examples~\ref{ex:embedded} and~\ref{ex:twisted});
    \item the preimage of $P$ is infinite (see Examples~\ref{ex:positive}~\Ra{and~\ref{ex:positive2}}).
\end{itemize}

\begin{Theorem}\label{thm:main}
    Consider polynomials $f_1, \ldots, f_n \in \mathbb{K}[\bx, \by]$  in variables $\bx = (x_1, \ldots, x_n)$ and $\by = (y_1, \ldots, y_m)$ homogeneous with respect to $\bx$ such that $\deg_{\bx}f_i > 0$ for every $1 \leqslant i \leqslant n$.
    Let $X$ be the scheme defined by $f_1, \ldots, f_n$ in $\PP^{n - 1}_{\overline{\KK}} \times \AAA^m_{\overline{\KK}}$, and we define $\pi \colon \PP^{n - 1}_{\overline{\KK}} \times \AAA^m_{\overline{\KK}} \to \AAA^m_{\overline{\KK}}$, the canonical projection.

    \Rb{Consider a point $P \in \AAA^m_{\overline{\KK}}$ such that
    \begin{enumerate}
        \item $\pi^{-1}(P) \cap X$ is finite and
        \item for every point $Q \in \pi^{-1}(P) \cap X$, $Q$ is a smooth point of a component of $X$ of dimension~$\geqslant m$.
    \end{enumerate}
    }
    Then $\PRes_{\bx}(f_1, \ldots, f_n)$ does not vanish at $P$.
\end{Theorem}

Furthermore, we provide a complete description of the zeroes of $\PRes_{\bx}(f_1, f_2)$ for the simplest case $n = 2$ and $m = 1$.

\begin{Proposition}\label{prop:planar}
    Consider $f_1, f_2 \in \KK[x_1, x_2, y]$ homogeneous with respect to $x_1, x_2$.
    We write $f_1 = g h_1$ and $f_2 = g h_2$, where $g = \operatorname{GCD}(f_1, f_2)$.
    Then $\PRes_{\bx}(f_1, f_2)$ vanishes at $y^\ast \in \overline{\KK}$ if an only if:
    \begin{itemize}
        \item $y^\ast$ is a $y$-coordinate of a solution of $h_1 = h_2 = 0$;
        \item $g$ is divisible by $y - y^\ast$ \Ra{over $\overline{\KK}$}.
    \end{itemize}
\end{Proposition}

\section{Examples}\label{sec:examples}

The computations for the examples below were carried out using the {\sc Resultants} package by~\cite{Staglian2018} in {\sc Macaulay2}~\citep{M2}. 
\Ra{The package computes multivariate resultant by induction in the number of variabes using the Poisson product formula~\cite[Section~1]{Staglian2018}.}

\begin{Example}[Proper embedded component]\label{ex:embedded}
    Let $n = 2$ and $m = 1$, and take 
    \[
    f_1 = x_1 (x_1 + x_2 y), \quad f_2 = x_1 (x_1 - x_2 y).
    \]
    The zero set $V(f_1, f_2)$ is the line $x_1 = 0$ but the ideal generated by $f_1, f_2$ has an embedded component (double point) at the origin (i.e., point $([0:1], 0)$).
    As described in Remark~\ref{rem:computation}, we compute that $\PRes_{\bx}(f_1, f_2) = y^2$, so the projection of the embedded component is in the zero set of the perturbed resultant.
    This agrees with Proposition~\ref{prop:planar}.
\end{Example}

\begin{Example}[Excess component with non-dominant projection]\label{ex:positive}
    Let $n = 2$, $m = 1$, and $f_1 = f_2 = x_1 y$.
    In this case the variety $V(f_1, f_2)$ will consist of two one-dimensional components $x_1 = 0$ and $y = 0$. 
    However, note that the latter has zero-dimensional image under the projection to the $y$-component.
    As described in Remark~\ref{rem:computation}, we compute $\PRes_{\bx}(f_1, f_2) = y$, and see that the projection of the excess component $y = 0$ is indeed in the zero set of $\PRes$.
    This agrees with Proposition~\ref{prop:planar}.
\end{Example}

\Ra{
\begin{Example}[Excess component with non-dominant projection]\label{ex:positive2}
    Consider the case $n = 3$ and $m = 1$ and take
    \[
    f_1 = yx_1 + x_2 + x_3,\quad f_2 = x_1 + yx_2 + x_3, \quad f_3 = x_1 + x_2 + y x_3.
    \]
    The zero set $V(f_1, f_2, f_3)$ consists of two components: a point $([1:1:1], -2)$ and a line $([a:b:-a-b], 1)$.
    Computing the perturbed resultant, we get
    \[
    \PRes_{\bx}(f_1, f_2, f_3) = (y - 1)^2(y + 2).
    \]
    Therefore, the excess component (the line) with non-dominant projection contributes to the zero set of $\PRes$.
\end{Example}}

\begin{Example}[Singularity, persistent]\label{ex:twisted}
    Let $n = 3$ and $m = 1$.
    We start with a rational curve in $\AAA^3$ parametrized by $s \to (s^3, s^4, s^5)$ in the coordinates $(y, x_1, x_2)$.
    Then we homogenize this curve with respect to $x_1, x_2$, and obtain a curve in $\PP^2\times \AAA$ defined by:
    \begin{equation}\label{eq:twisted}
    f_1 = x_3^2 y^3 - x_1 x_2,\quad f_2 = x_1^2 - x_2 x_3 y, \quad f_3 = x_2^2 - x_1 x_3 y^2.
    \end{equation}
    These three polynomials define a prime ideal of a curve which projects dominantly on the $y$-coordinate.
    However, this curve has a singularity at the origin (i.e., the point $([0:0:1], 0)$) where it is not a locally complete intersection.
    We compute $\PRes_{\bx}(f_1, f_2, f_3)$ as described in Remark~\ref{rem:computation}, and obtain 
    \[
    \PRes_{\bx}(f_1, f_2, f_3) = y.
    \]
    Note that the fact that $\PRes_\bx(f_1, f_2, f_3)$ may only vanish at $y = 0$ follows from Theorem~\ref{thm:main}.
\end{Example}

\begin{Example}[Singularity, not persistent]\label{ex:sing_not_pers}
   Let us show that not all the singularities yield persistent components.
   We consider the case $n = 2, m = 1$ and
   \[
   f_1 = f_2 = x_2^2 y^3 - x_1^2.
   \]
   The variety $V(f_1, f_2)$ is a curve with a singularity at $([0:1], 0)$.
   However, direct computation, as described in Remark~\ref{rem:computation}, shows that $\PRes_{\bx}(f_1, f_2) = 1$ (this could also be deduced from Proposition~\ref{prop:planar}), so the singularity is not reflected in the perturbed resultant.
   We think that the main difference between this example and Example~\ref{ex:twisted} is that in the latter, the curve is not a locally complete intersection at the singularity. This property forces all the syzygies of~\eqref{eq:twisted} to vanish at the singular point \Rb{(see Remark~\ref{rem:Sp} for further details)}.
\end{Example}

\Ra{
\begin{Example}[Singularity, not persistent]
    The previous example featured a curve with a cusp singularity, here we will consider a space curve with a self-intersection.
    We take $n = 3, m = 1$ and
    \[
    f_1 = x_1^2 + x_2^2 - x_3^2(1 - y^2), \quad f_2 = x_1^2 + x_2^2 - x_1 x_3, \quad f_3 = f_1 + f_2.
    \]
    These equations define a curve with a self-intersection at $([1:0:1], 0)$ but this point does not make a contribution to the zero set of the perturbed resultant: direct computation shows that $\PRes_{\bx}(f_1, f_2, f_3) = 1$.
\end{Example}}


\section{Proofs}\label{sec:proofs}
In this section all the varieties will be considered over the algebraic closure $\overline{\KK}$ of $\KK$ and $\PP$ and $\AAA$ will always mean $\PP_{\overline{\KK}}$ and $\AAA_{\overline{\KK}}$.
Throughout the section, we fix 
\begin{itemize}
    \item two sets of variables $\bx = (x_1, \ldots, x_n)$ and $\by = (y_1, \ldots, y_m)$;
    \item subalgebra $R \subset \KK[\bx, \by]$ of polynomials homogeneous {\Ra{with respect to}} $\bx$;
    \item the canonical projection $\pi\colon \PP^{n - 1}\times \AAA^m \to \AAA^m$ induced by the embedding $\KK[\by] \subset R$;
    \item polynomials $f_1, \ldots, f_n \in R$ with degrees $d_i := \deg_{\bx} f_i$, and we will assume that $d_1, \ldots, d_n$ are positive;
    \item space $\Ra{\mathcal{P}}$ of all vectors $\bp \in (\KK[\bx])^n$ of homogeneous polynomials with $\deg_\bx p_i = d_i$ for $1 \leqslant i \leqslant n$.
\end{itemize}
Our study of the perturbed resultant $\PRes_{\bx}(f_1, \ldots, f_n)$ will proceed via perturbing the whole variety $X := V(f_1, \ldots, f_n)$ defined by $f_1, \ldots, f_n$.
More precisely, let $\bp \in \Ra{\mathcal{P}}$ be an admissible vector of homogeneous polynomials.
As in Definition~\ref{def:canny}, we set $f_{i, \bp} := f_i + \varepsilon p_i$ for $1 \leqslant i \leqslant n$.
Consider the variety 
\[
\widetilde{X}_{\bp} := V(f_{1, \bp}, \ldots, f_{n, \bp}) \subset \PP^{n - 1} \times \AAA^m \times \AAA.
\]
It has a decomposition $\widetilde{X}_{\bp} = Y_{\bp} \cup Z_{\bp}$, where $Y_{\bp}$ is the union of all components of $\widetilde{X}_{\bp}$ contained in $\{ \varepsilon = 0\}$, 
and $Z_{\bp}$ is the union of the remaining components.
We define the perturbed variety $X_{\bp} \subset \PP^{n - 1} \times \AAA^m$ as $Z_{\bp} \cap \{ \varepsilon = 0 \}$.
The following lemma justifies the relevance of this construction to our problem.

\begin{Lemma}\label{lem:Xp}	
    For every $\bp \in \Ra{\mathcal{P}}$, $\pi(X_{\bp}) = V\left( \CRes_{\bx, \bp}(f_1, \ldots, f_n) \right)$.
\end{Lemma}

\begin{proof}
	Let $\widetilde{\pi} \colon \PP^{n - 1} \times \AAA^m \times \AAA  \to \AAA^m \times \AAA$ be the canonical projection.
	From the general property of the resultant, we know that 
    \[
    \widetilde{\pi} \left( \widetilde{X}_{\bp} \right) = V\left( \Res_{\mathbf{x}} (f_{1, \bp}, \ldots, f_{n, \bp}) \right).
    \]
    \Rb{Furthermore, we have $\widetilde{\pi}(\widetilde{X}_{\bp}) = \widetilde{\pi}(Y_\bp \cup Z_{\bp}) = \widetilde{\pi}(Y_{\bp}) \cup \widetilde{\pi}(Z_{\bp})$.
    We factorize
    \[
    \Res_{\bx}(f_{1, \bp}, \ldots, f_{n, \bp}) = \varepsilon^s (\CRes_{\bx, \bp}(f_1, \ldots, f_n) + \mathrm{o}(\varepsilon)).
    \]
    Since $\widetilde{\pi}(Y_{\bp}) \subset V(\varepsilon)$ and none of the components of $\widetilde{\pi}(Z_{\bp})$ belongs to $V(\varepsilon)$, we have
    \[
    \widetilde{\pi}(Z_{\bp}) = V(\CRes_{\bx, \bp}(f_1, \ldots, f_n) + \mathrm{o}(\varepsilon)).
    \]
    Since $\widetilde{\pi}(X_{\bp}) = \widetilde{\pi}(Z_{\bp} \cap \{\varepsilon = 0\}) = \widetilde{\pi}(Z_{\bp}) \cap \{\varepsilon = 0\}$, we have
    \[
    \widetilde{\pi}(X_{\bp}) = V(\varepsilon, \CRes_{\bx, \bp}(f_1, \ldots, f_n) + \mathrm{o}(\varepsilon)) = V(\varepsilon, \CRes_{\bx, \bp}(f_1, \ldots, f_n)).
    \]
    Therefore, $\pi(X_{\bp}) = V(\CRes_{\bx, \bp}(f_1, \ldots, f_n))$.}
\end{proof}

Now we can concentrate on studying $X_{\bp}$.
For any admissible $\bp \in (\KK[\bx])^n$, we define
\[
S_{\bp} := \{ g_1p_1 + \ldots + g_n p_n \mid (g_1, \ldots, g_n) \in \operatorname{Syz}(f_1, \ldots, f_n) \},
\]
where $\operatorname{Syz}(f_1, \ldots, f_n)$ denotes the module of syzygies of $f_1, \ldots, f_n$.

\begin{Lemma}\label{lem:syzygy}
    In the notation above, $X_{\bp} \subset V(f_1, \ldots, f_n, S_{\bp})$.
    Furthermore, if none of the components of $V(S_{\bp})$ is contained in $V(f_1, \ldots, f_n)$, then $X_{\bp} = V(f_1, \ldots, f_n, S_{\bp})$.
\end{Lemma}

\begin{Remark}\label{rem:Sp}
    We did not find an example in which a component of $V(S_{\bp})$ would be contained in $V(f_1, \ldots, f_n)$ for generic $\bp$, so we conjecture that this never happens.
    Being established, this would allow to strengthen our main result \Rb{as follows}.
    \Rb{At generic $\bp$, we would have $X_{\bp} = V(f_1, \ldots, f_n, S_{\bp})$, so taking sufficiently many generic $\bp$'s, $X_{\bp_1}\cap \ldots \cap X_{\bp_\ell}$ will be exactly the locus of $V(f_1, \ldots, f_n)$ where all the components of all syzygies vanish.
    This would allow to rule out some of the singularities (for example, as in Example~\ref{ex:sing_not_pers}).}
\end{Remark}

\begin{proof}[Proof of Lemma~\ref{lem:syzygy}]
    Throughout the proof, for a point $P \in \PP^{n - 1} \times \AAA^m$ and a scalar $\alpha \in \KK$, we will denote by $(P, \alpha)$ the point $\PP^{n - 1} \times \AAA^m \times \AAA$ with the $\varepsilon$-coordinate equal to $\alpha$.

    By the construction of $X_\bp$, we have $X_{\bp} \subset V(f_1, \ldots, f_n)$.
    We will now show that $X_{\bp} \subset V(S_{\bp})$.
    Consider a point $P \in X_{\bp}$.
    By definition of $X_{\bp}$ and $Z_{\bp}$, there exists an algebraic curve $\mathcal{C} \subset Z_{\bp}$ such that $P \in \mathcal{C}$ and $\mathcal{C}$ does not belong to $\{\varepsilon = 0\}$.
    Let $(g_1, \ldots, g_n)$ be any syzygy of $(f_1, \ldots, f_n)$.
    Then
    \[
    I(Z_{\bp}) \ni g_1 f_{1, {\bp}} + \ldots + g_n f_{n, \bp} = \varepsilon(g_1 p_1 + \ldots + g_n p_n).
    \]
    Since $\varepsilon$ is not identically zero on $\mathcal{C}$, the polynomial $g_1 p_1 + \ldots + g_n p_n$ vanishes on $\mathcal{C}$ and, thus, it vanishes at $P$.
    Therefore, $X_{\bp} \subset V(f_1, \ldots, f_n, S_{\bp})$.

    Now we will prove the reverse inclusion $X_{\bp} \supset V(f_1, \ldots, f_n, S_{\bp})$ for the case when none of the components of $V(S_{\bp})$ belongs to $V(f_1, \ldots, f_n)$.
    Consider any point $P \in V(f_1, \ldots, f_n, S_{\bp})$.
    It belongs to a component of $V(S_{\bp})$, we denote this component by $Y$.
    Since $Y$ does not belong to $V(f_1, \ldots, f_n)$, there exists $i$ such that $Y \not\subset V(f_i)$.
    Without loss of generality, we will assume $i = 1$.
    Therefore, there exists an irreducible curve $\mathcal{C} \subset Y$ such that $P \in \mathcal{C}$ and $\mathcal{C} \not\subset V(f_1)$.
    Due to the relation $f_j \cdot f_1 + (-f_1) \cdot f_j = 0$, $S_\bp$ contains a polynomial $f_j p_1 - f_1 p_j$ for every $2 \leqslant j \leqslant n$.
    Therefore, 
    \[
    (f_1(Q), p_1(Q)),\; \ldots,\; (f_n(Q), p_n(Q))
    \]
    are linearly dependent for every $Q \in \mathcal{C}$.
    Since the system $\bp = \mathbf{0}$ does not have a solution in $\PP^{n - 1}$, at least one of these vectors is nonzero, so, for every $Q \in \mathcal{C}$, there exists a unique $\varepsilon_Q$ such that $(Q, \varepsilon_Q) \in \widetilde{X}_{\bp}$. 
    This provides us with a bijective lifting of $\mathcal{C}$ to a curve $\mathcal{C}_0 \subset \widetilde{X}_{\bp}$.
    Since $f_1$ does not vanish identically on $\mathcal{C}$, $\varepsilon_Q \neq 0$ on an open subset of $\mathcal{C}_0$, so, by the irreducibility of $\mathcal{C}$ and, thus, $\mathcal{C}_0$, we have $\mathcal{C}_0 \subset Z_{\bp}$.
    Since $(P, 0) \in \mathcal{C}_0$, we deduce that $P \in X_{\bp}$.
\end{proof}


\begin{Proposition}\label{prop:smooth}
    Consider a smooth point $P$ on a component $Y \subset V(f_1, \ldots, f_n)$ (considered as an affine scheme) of dimension $\geqslant m$.
    Then there is an open nonempty subset $U \subset \Ra{\mathcal{P}}$
    such that, for every $\bp \in U$, we have $P \not\in X_\bp$.
\end{Proposition}
    
\begin{proof}
   Consider the local ring $(\mathcal{O}, \mathfrak{m})$ of the point $P$ considered as a point of the ambient space $\PP^{n - 1} \times \AAA^m$.
   We will denote the images of $f_1, \ldots, f_n$ in $\mathcal{O}$ by the same letters.
   Let $d := \dim Y$.
   The smoothness of $P$ implies that the Jacobian of $f_1, \ldots, f_n$ has rank $c := m + n - 1 - d$ at $P$. 
   By reordering $f$'s if necessary, we will further assume that the rank of the Jacobian of $f_1, \ldots, f_c$ is equal to $c$.
   Since $P$ is a smooth point, $\mathcal{O} / \langle f_1, \ldots, f_n \rangle \mathcal{O}$ is a regular local ring, so~\cite[proof of Proposition~1.10, Chapter VI]{kunz} implies that $\langle f_1, \ldots, f_n\rangle \mathcal{O} = \langle f_1, \ldots, f_c\rangle\mathcal{O}$.
   Since $d \geqslant m$, we have $c < n$, so we can write
   \[
     f_n = a_1 f_1 + \ldots + a_{n - 1}f_{n - 1},
   \]
   where $a_1, \ldots, a_{n - 1} \in \mathcal{O}$.
   By clearing denominators in this equality, we obtain a syzygy
   \[
     b_1 f_1 + \ldots + b_n f_n = 0,
   \]
   where $b_1, \ldots, b_n \in R$ and $b_n(P) \neq 0$.
   Then the condition
   \[
     b_1(P) p_1(P) + \ldots + b_n(P) p_n(P) \neq 0
   \]
   defines an open subset $U \subset \Ra{\mathcal{P}}$.
   In order to show that this subset is not empty, we choose linearly independent linear forms $\ell_1, \ldots, \ell_n$ such that $\ell_1(P) = \ldots = \ell_{n - 1}(P) = 0$ and $\ell_n(P) = 1$.
   Then $(\ell_1^{d_1}, \ldots, \ell_n^{d_n}) \in U$.
   By Lemma~\ref{lem:syzygy}, $P \not\in X_{\bp}$ for every $\bp \in U$.
\end{proof}

\begin{Proposition}\label{prop:main}
    Consider a point $P \in \AAA^m$ such that $\pi^{-1}(P)$ is finite and every point $Q \in \pi^{-1}(P)$ is a smooth point of an excess component of the scheme defined by $f_1, \ldots, f_n$.
    Then $\PRes_{\bx}(f_1, \ldots, f_n)$ does not vanish at $P$.
\end{Proposition}

\begin{proof}
    Let $U_0$ be the Zariski open subset of $\Ra{\mathcal{P}} \times \Ra{\mathcal{P}}$ from Lemma~\ref{lem:double_pert} and $U_1$ be the intersection of the open subsets of $\Ra{\mathcal{P}}$ provided by Proposition~\ref{prop:smooth} applied to the elements of $\pi^{-1}(P)$.
    Let $(\bp, \bq) \in (U_0 \cap (U_1 \times U_1))$.
    By Lemma~\ref{lem:Xp}, we have
    \[
    V(\PRes_\bx(\mathbf{f})) \subset V(\CRes_{\bx, \bp}(\mathbf{f})) \cap V(\CRes_{\bx, \bq}(\mathbf{f})) = \pi(X_{\bp}) \cap \pi(X_{\bq}).
    \]
    Since $\bp, \bq \in U_1$, the latter intersection does not contain $P$, so $\PRes_\bx(\mathbf{f})$ does not vanish on~$P$.
\end{proof}

\begin{proof}[Proof of Theorem~\ref{thm:main}]
    Follows directly from Proposition~\ref{prop:main}.
\end{proof}

\begin{proof}[Proof of Proposition~\ref{prop:planar}]
    We start with observing that the syzygy module of $f_1, f_2$ is generated by $(h_2, -h_1)$.
    Therefore, $S_{\bp}$ is generated by a single polynomial $s_{\bp} := p_1h_2 - p_2h_1$. 
    Since, for every irreducible divisor $g_0$ of $g$, at least one of $h_1$ and $h_2$ is not divisible by $g_0$,
    there exists an nonempty open subset $U_0 \subset \Ra{\mathcal{P}}$ such that $s_{\bp}$ is coprime with $g$ for $\bp \in U_0$.
    Then Lemma~\ref{lem:syzygy} implies that 
    \[
    X_{\bp} = V(f_1, f_2, s_{\bp})\quad \text{ for }\bp \in U_0.
    \]

    \noindent
    \Ra{\textbf{Case 1:}} \emph{Consider $y^\ast \in \overline{\KK}$ such that $g$ is divisible by $y - y^\ast$}.
    Then 
    \[
    V(f_1, f_2, s_{\bp}) \cap \{y = y^\ast\} = V(s_{\bp}) \cap \{y = y^\ast\}
    \]
    is always nonempty because $s_{\bp}(x_0, x_1, y^\ast)$ will always have a root in $\PP^1$.
    Consider $\bp \in U_0$.
    Then $X_{\bp} \cap \{y = y^\ast\}$ is nonempty, so, by Lemma~\ref{lem:Xp}, $\CRes_{\bx, \bp}(f_1, f_2)$ vanishes at $y^\ast$.
    Then the same is true for $\PRes_\bx(f_1, f_2)$. 

    \noindent
    \Ra{\textbf{Case 2:}} \emph{Consider $y^\ast \in \overline{\KK}$ such that there exist $\bx^\ast \in \PP^1$ with} 
    \[
    h_1(\bx^\ast, y^\ast) = h_2(\bx^\ast, y^\ast) = 0.
    \]
    Consider $\bp \in U_0$.
    Then $s_{\bp}$ vanishes at $(\bx^\ast, y^\ast)$, so $(\bx^\ast, y^\ast) \in X_{\bp}$.
    Then, by Lemma~\ref{lem:Xp}, $\CRes_{\bx, \bp}(f_1, f_2)$ vanishes at $y^\ast$.
    Then the same is true for  $\PRes_\bx(f_1, f_2)$ as well.

    \noindent
    \Ra{\textbf{Case 3:}} \emph{Finally, consider $y^\ast$ which does not fall into neither of the two cases above.}
    Let $\bx_1^\ast, \ldots, \bx_\ell^\ast \in \PP^1$ be the points such that $(\bx_i^\ast, y^\ast)$ is a solution of $f_1 = f_2 = 0$ for every $1 \leqslant i \leqslant \ell$.
    Since $g$ is not divisible by $y - y^\ast$, there are only finitely many such points.
    We take $\bp$ with undetermined coefficients and 
     consider the condition 
    \begin{equation}\label{eq:cond}
    \prod\limits_{i = 1}^\ell (p_1(\bx_i^\ast) h_1(\bx_i^\ast, y^\ast) - p_2(\bx_i^\ast)h_2(\bx_i^\ast, y^\ast)) \neq 0.
    \end{equation}
    By the assumption on $y^\ast$, none of the brackets in the product vanishes identically, so~\eqref{eq:cond} defines a nonempty open subset $U_1 \subset \Ra{\mathcal{P}}$.
    For every $\bp \in U_0 \cap U_1$, none of $(\bx_i^\ast, y^\ast)$ belongs to $X_{\bp}$, so, by Lemma~\ref{lem:Xp}, $\CRes_{\bx, \bp}(f_1, f_2)$ does not vanish on $y^\ast$. 
    Therefore, $\PRes_\bx(f_1, f_2)$ does not vanish on $y^\ast$ as well.
\end{proof}

\bibliographystyle{elsarticle-harv} 
\bibliography{bib}


\end{document}